\documentclass[11pt]{amsart}

\usepackage{amsaddr}

\usepackage{amsthm,amssymb,amsmath}

\newtheorem{theorem}{Theorem}
\newtheorem{lemma}[theorem]{Lemma}
\newtheorem{corollary}[theorem]{Corollary}

\usepackage{microtype}

\title{Consensus Game Acceptors and Iterated Transductions}

\author[Dietmar Berwanger \and Marie van den Bogaard]{\smaller Dietmar Berwanger \and Marie van den Bogaard}
\address{{\smaller LSV, CNRS and Universit\'e Paris-Saclay, France}}
\email{dwb@lsv.fr (Dietmar Berwanger, corresponding author)}
\usepackage{enumerate}
\usepackage{subfig}

\usepackage{url}
\usepackage{tikz}

\usetikzlibrary{arrows,topaths}
\usetikzlibrary{shapes}
\usetikzlibrary{fit,matrix}
\usetikzlibrary{automata, positioning}

\newtheorem{proposition}{Proposition}[section]

\newcommand{\PSPACE}{\textsc{PSpace}}

\newcommand*{\sBox}{\mbox{\scriptsize $\Box$}}

\newcommand*{\acc}{\mathrm{acc}}
\newcommand*{\rej}{\mathrm{rej}}

\newcommand*{\calA}{\mathcal{A}}
\newcommand*{\calS}{\mathcal{S}}
\newcommand*{\calB}{\mathcal{B}}
\newcommand*{\calR}{\mathcal{R}}

\renewcommand*{\Lambda}{A}
\newcommand*{\DA}{D_{{\kern-1.5pt}\Lambda}}
\newcommand*{\DGamma}{D_{{\kern0pt}\Gamma}}

\newcommand*{\noL}{\boxed{\scriptstyle (}}
\newcommand*{\noR}{\boxed{\scriptstyle )}}
\newcommand*{\noLL}{\boxed{\scriptstyle [}}
\newcommand*{\noRR}{\boxed{\scriptstyle ]}}
\newcommand*{\Boxed}[1]{\boxed{\scriptstyle #1}}

\def\draftnote{\today\quad\currenttime\quad IJFCS: DLT 2016 \qquad\jobname}%

\begin{document}

\thispagestyle{empty}






\setlength{\textfloatsep}{5pt}

\begin{abstract}

We study a game for recognising formal languages, 
in which two players with imperfect information 
need to coordinate on a common decision, given private
input words correlated by a finite graph.
The players have a joint objective to avoid an inadmissible
decision, in spite of the uncertainty induced by the input.

We show that the acceptor model based on consensus games 
characterises context-sensitive languages. 
Further, we describe the expressiveness of these games in terms of 
iterated synchronous transductions 
and identify a subclass that characterises context-free languages. 

\end{abstract}

\maketitle


\section{Introduction}

The idea of viewing computation as an interactive process has been at
the origin of many enlightening developments over the past three decades.
With the concept of alternation, 
introduced around 1980 by Chandra and Stockmeyer, and independently
by Kozen~\cite{ChandraKS81}, computation steps are attributed to
conflicting players seeking to reach or avoid certain outcome
states.
This approach relies on determined
games with perfect information, and it led
to important and elegant results, particularly in automata theory. 
Around the same time, Peterson and Reif~\cite{PetersonRei79} 
initiated a study on computation via games
with imperfect information, involving teams of players. 
This setting turned out to be highly expressive, but also 
overwhelmingly difficult to
comprehend. (See~\cite{AzharPetRei01,HearnDem09}, for more
recent accounts.)

In this paper, we propose a game model of a language acceptor
based on coordination games between two 
players with imperfect information. Compared to the model
of Reif and Peterson, our setting is utterly simple: the games are
played on a finite graph, plays are of finite duration, they
involve only one yes-or-no decision, and the players have no
means to communicate. Moreover, they are bound to take their decisions
in consensus. Given an input word that may yield different observations to each
of the players, they have to settle simultaneously and
independently on a common decision, otherwise they lose. 

Consensus game acceptors arise as a 
particular case of
coordination games with perfect recall, also described as 
multiplayer concurrent games or
synchronous distributed games with incomplete information 
in the computer-science literature.
Our motivation for studying the acceptor model is to obtain
lower bounds on the complexity of basic computational problems 
on coordination games with imperfect information, specifically
(1) solvability: whether a
winning strategy exists for a given game, and (2) implementability: 
which computational resources are needed to implement a winning
strategy, if any exists. 

Without the restrictions to consensus and to a single decision per play, 
the solvability problem for
coordination games with safety winning conditions is known to be
undecidable~\cite{PetersonRei79,PnueliRos90}. 
Furthermore, Janin~\cite{Janin07} points out that there exist two-player safety games 
that admit a winning strategy but none that can be implemented by a Turing
machine. 

Our first result establishes a
correspondence between context-sensitive languages and
consensus games: 
We prove that, for every context-sensitive language~$L$, there exists
a solvable consensus game in which every winning strategy extends the
characteristic function of~$L$, and conversely, that
every solvable consensus game admits a winning strategy 
characterised by a context-sensitive language. 

As a second result, we characterise winning strategies for consensus
games in terms of iterated transductions of the (synchronous rational) 
relation between the observations of players. This allows us to  
identify a subclass of games that corresponds to
context-free languages. Although it is still undecidable whether a game of
the class admits a winning strategy, we 
can effectively construct optimal strategies 
implemented by push-down automata. 

The results provide insight on the inherent complexity of 
coordination in games with imperfect information. 
With regard to the basic problem of 
agreement on a simultaneous action, 
they substantiate the assertion
that ``optimality requires computing common knowledge'' 
put forward by   
Dwork and Moses in their analysis of Byzantine 
agreement in distributed systems~\cite{DworkMos90}. 
Indeed, the constraints induced by our acceptor model 
can be reproduced in virtually any kind of
games with imperfect information and plays of unbounded length,
with the consequence that 
implementing optimal strategies amounts to 
deciding the transitive closure of  
the transduction induced by the game graph.

\subsection*{Acknowledgements} 
This work was partially supported by CASSTING Project of the 
European Union Seventh Framework Programme.
The current paper extends a preliminary report~\cite{BerwangerB15} presented at the Conference on Developments in Language Theory (DLT 2015).

\section{Preliminaries}

For standard background on formal language theory, in particular
context-sensitive
languages, we refer to Chapters~3 and~4 of the handbook~\cite{HFL1}. 
We use the characterisation of context-sensitive
languages in terms of nondeterministic linear-bounded automata given by
Kuroda~\cite{Kuroda64}, and the following well-known results from the same 
article: (1) For a fixed context-sensitive language~$L$ over an
alphabet~$\Sigma$,  the problem whether a given word $w \in \Sigma^*$
belongs to~$L$ is $\PSPACE$-hard. (2) The problem of determining whether a
given context-sensitive language represented by a linear-bounded
automaton contains any non-empty word is undecidable. 

\subsection{Consensus game acceptors}

Consensus acceptors are games between two cooperating players $1$ and $2$, 
and a passive agent called Input.  
Given a finite \emph{observation alphabet} $\Gamma$ common to both players,
a~\emph{consensus game acceptor} $G = (V, E, (\beta^1, \beta^2), v_0, \Omega )$ 
is described by a finite set $V$ of \emph{states}, 
a \emph{transition} relation $E \subseteq V \times V$,
and a pair of \emph{observation} functions $\beta^i: V \to \Gamma$
that label every state with an observation, for each player~$i = 1,2$.  
There is a distinguished \emph{initial} state $v_0 \in V$ with no
incoming transition. 
States with no outgoing transitions are called \emph{final}
states; the admissibility condition $\Omega: V \to \mathcal{P}(
\{ 0, 1 \})$ maps every final state~$v \in V$ 
to a nonempty subset of admissible decisions $\Omega( v ) \subseteq \{
0, 1\}$. The observations at the initial and the final states do not
matter, we assume that they correspond to a special symbol
$\# \in \Gamma$ for both players.

The game is played as follows: Input chooses a finite path
$\pi = v_0 \, v_1 \, \dots \, v_{n + 1}$ in $G$ 
from the initial state~$v_0$, following
transitions
$(v_{\ell}, v_{\ell + 1}) \in E$, for all 
$\ell \le n$, to a final state~$v_{n+1}$. 
Then, each player~$i$ receives a private sequence of observations
$\beta^i( \pi ) := \beta^i (v_1) \, \beta^i (v_2)\, \dots\, \beta^i(v_n)$
and is asked to take a \emph{decision} $a^i \in \{ 0, 1 \}$, 
independently and simultaneously. 
The players win if they agree on an admissible decision, that is, 
$a^1 = a^2 \in \Omega( v_{n+1})$; otherwise they lose. 
Without risk of confusion we sometimes write $\Omega( \pi )$ for
$\Omega( v_{n+1} )$. 
 
We say that 
two plays $\pi, \pi'$ are \emph{indistinguishable} to player~$i$, and write 
$\pi \sim^i \pi'$, if $\beta^i(\pi) = \beta^i(\pi')$. 
This is an equivalence relation, and its classes,
called the \emph{information sets} of Player~$i$, 
correspond to observation sequences $\beta^i( \pi )$. 
A \emph{strategy} for Player~$i$ is a mapping $s^i: V^* \to \{ 0, 1 \}$
from plays $\pi$ to decisions $s^i( \pi ) \in \{0, 1\}$
such that  $s^i ( \pi ) = s^i( \pi')$, 
for any pair~$\pi \sim^i \pi'$ 
of indistinguishable plays.
A joint strategy is a pair $s = (s^1, s^2)$; it is \emph{winning}, if 
$s^1( \pi ) = s^2( \pi ) \in \Omega( \pi )$, for all plays $\pi$. 
In this case, the components $s^1$ and $s^2$ are equal, 
and we use the term \emph{winning strategy} to refer to the
joint strategy
or either of its components.
Finally, a game is \emph{solvable}, if there exists a (joint) winning
strategy.

In the terminology of distributed systems, 
consensus game acceptors correspond to \emph{synchronous} systems 
with \emph{perfect recall} and \emph{known initial state}. 
They are a particular case of distributed games
with safety objectives \cite{MohalikWal2003}, coordination games with
imperfect information~\cite{Berwanger10}, or multi-player concurrent
games~\cite{AlurHenKup02}. Whenever we refer to games
in the following, we mean consensus game acceptors. 

\subsection{Strategies and knowledge}
We say that two plays $\pi$ and $\pi'$
are \emph{connected}, and write $\pi \sim^* \pi'$, 
if there exists a sequence of  
plays $\pi_1, \dots, \pi_{k}$ such that 
$\pi \sim^{1} \pi_1 \sim^{2} \dots \sim^{1} \pi_{k}  \sim^{2} \pi'.$
Then, a mapping $f: V^* \to \{0, 1\}$ from plays to decisions 
is a strategy that satisfies the consensus condition 
if, and only if, $f( \pi ) = f( \pi' )$, for
all $\pi \sim^* \pi'$. In terms of distributed knowledge, 
this means that the decisions have to be based on events
that are
common knowledge among the players at every play. 
(For an introduction to knowledge in distributed systems, 
see Chapters 10 -- 11 in the book of Fagin, Halpern, Moses, and Vardi~\cite{FaginHMV95}.)
Such a consensus strategy --- or, more precisely, the pair
$(f,f)$--- 
may still fail, due to prescribing
inadmissible decisions. We say that a decision $a \in \{0, 1\}$
is \emph{safe} at a play $\pi$ if $a \in \Omega( \pi' )$, for all
$\pi' \sim^* \pi$. Then, a consensus strategy $f$ is
winning, if and only if, it prescribes a safe decision $f( \pi )$, 
for every play $\pi$.

It is sometimes convenient to represent a strategy for a player~$i$ 
as a function $f^i: \Gamma^* \to \{ 0, 1 \}$. Every such function 
describes a valid strategy, 
because observation sequences identify 
information sets; we refer to an \emph{observation-based} strategy in contrast to the
\emph{state-based} representation $s^i: V^* \to \{0, 1\}$.
Note that the components of a joint winning strategy need no longer
be equal in the
observation-based representation. 
However, once the
strategy for one player is fixed, the strategy of the other player is
determined by the consensus condition, so there is no risk of
confusion in speaking of a winning strategy rather than
a joint strategy pair.

As an example, consider the game depicted in
Figure~\ref{fig:game-cover}, with observation 
alphabet $\Gamma=\{ a, b, \triangleleft, \triangleright,
\mbox{\scriptsize $\Box$}\}$. 
States $v$ at which the two players receive different observations are
split, with $\beta^1( v )$ written in the upper part and $\beta^2(
v )$ in the lower part; states at which the players receive the same
observation carry only one symbol. 
The admissible decisions at final states 
are indicated on the outgoing arrow. 
Notice that upon receiving the
observation sequence $a^2b^2$, for instance,
the first player is constrained to choose decision $1$, due to the
following sequence of indistinguishable plays that leads
to a play where deciding $0$ is not admissible.
\begin{align*}
\begin{pmatrix}
  a, a\\
  a, \triangleleft\\
  b, \triangleright\\
  b, b
\end{pmatrix} 
\sim^2
\begin{pmatrix}
  a, a\\
  \triangleleft, \triangleleft\\
  \triangleright, \triangleright\\
  b, b\\
\end{pmatrix} 
\sim^1
\begin{pmatrix} 
  a, \triangleleft\\
  \triangleleft, \triangleright\\
  \triangleright, \triangleleft\\
  b, \triangleright
\end{pmatrix}
\sim^2 
\begin{pmatrix}
  \triangleleft, \triangleleft \\
  \triangleright, \triangleright \\
  \triangleleft, \triangleleft \\
  \triangleright, \triangleright \\
\end{pmatrix}
\sim^1 
\begin{pmatrix}
  \triangleleft, \sBox\\ 
  \triangleright, \sBox \\
  \triangleleft, \sBox \\
  \triangleright, \sBox 
\end{pmatrix}
\sim^2 
\begin{pmatrix}
  \sBox, \sBox\\ 
  \sBox, \sBox \\
  \sBox, \sBox \\
  \sBox, \sBox 
\end{pmatrix}
\end{align*}
In contrast, decision~$0$ may be safe when Player~$1$ receives input
$a^3b^2$, for instance. Actually, the strategy $s^1(w )$ that
prescribes~$1$ if, and only if, $w \in \{a^nb^n~|~n \in
\mathbb{N}\}$ determines a joint winning strategy. 
Next, we shall make the relation between games and languages more precise.

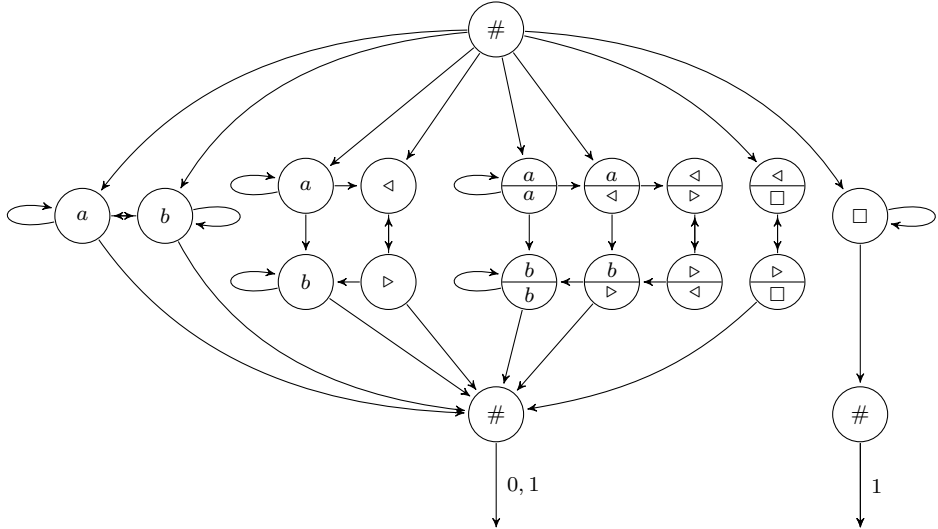
\begin{figure}[t]
\tikzset{every picture/.style={node distance=5em,->,>=stealth',shorten >=1pt,auto, font=\scriptsize}} 
\tikzset{initial text=, initial distance=1em} 
\tikzset{accepting by double} 
\tikzstyle{every state}=[minimum size=1.9em, inner sep=1.8pt, font=\scriptsize]

\begin{center}
   
      \begin{tikzpicture}[xscale=1.1,yscale=.8]
      
        \node[state]                (0)     at (5,8.4)         {$\#$};

        \node[state]              (A01) at (2.7,5.8) {$a$};
        \node[state]              (A11) at (3.7,5.8) {$\triangleleft$};
        \node[state]              (c01) at (0,5.3) {$a$};
         \node[state]              (c02) at (1,5.3) {$b$};
        
        \node[state with output]             (A20) at (5.4,5.8) {$a$ \nodepart{lower} $a$}; 
        \node[state with output]             (A21) at (6.4,5.8) {$a$ \nodepart{lower} $\triangleleft$};  
        \node[state with output]             (A22) at (7.4,5.8) {$\triangleleft$ \nodepart{lower} $\triangleright$};

        \node[state with output]            (A3) at (8.4,5.8) {$\triangleleft$ \nodepart{lower} {\tiny $\Box$}};
        
       \node[state]              (C00) at (9.4,5.3) {{\tiny $\Box$}};

        \node[state]              (B11) at (3.7,4.2) {$\triangleright$};
        \node[state]              (B01) at (2.7,4.2) {$b$};

        \node[state with output]             (B20) at (5.4,4.2) {$b$ \nodepart{lower} $b$};
        \node[state with output]             (B21) at (6.4,4.2) {$b$ \nodepart{lower} $\triangleright$};
        \node[state with output]              (B22) at (7.4,4.2) {$\triangleright$  \nodepart{lower}  $\triangleleft$};
        
         \node[state with output]            (B3) at (8.4,4.2) {$\triangleright$ \nodepart{lower} {\tiny $\Box$}};

        \node[state]                (BB) at (5,2) {$\#$};   
        \node[state, draw=none]     (Z) [below of=BB] {};
        \node[state]              (B00) at (9.4,2) {$\#$};
       \node[state, draw=none]                (AA) [below of=B00] {};

        \path (0) edge [bend right =30] (c01)
        edge [bend right =30] (c02)
        edge(A01)
        edge (A11)
        edge (A20)
        edge (A21)
        edge [bend left =20] (A3)
        edge [bend left =30] (C00)

       (A01) edge (B01)
       (A11) edge (B11)
       (A20) edge  (B20)
       (A21) edge  (B21)
       (A22) edge  (B22)
       (A3) edge (B3)

        (c01) edge [bend right =30] (BB)
        (c02) edge [bend right =30] (BB) 
        (B00) edge (AA)
        (B01) edge (BB)
        (B11) edge (A11)
        (B11) edge (BB)
        (C00) edge  (B00)
         
       (B20) edge  (BB)
       (B21) edge  (BB)
      
       (B3) edge (A3)
       (B3) edge [bend left =20] (BB);
       
        \path (B00) 
        edge node {$1$} 
        (AA);
h        
        \path (BB) 
        edge node[right] {$0,1$} 
        (Z);
        
         \path (A01) 
        edge (A11)
       
       (B11) edge (B01)
        
        (A20) edge (A21)
        (A21) edge (A22)
        (c01) edge (c02)
        
        (c02) edge (c01)
        (B22) edge (B21)
        (B21) edge (B20)
        (B22) edge (A22)
        ;
        \path (c01) edge [loop left] node[above] {} (c01);
        \path (c02) edge [loop right] node[above] {} (c02);
        \path (A01) edge [loop left] node[above] {} (A01);
       \path (C00) edge [loop right] node[above] {} (C00);
        \path (B01) edge [loop left] node[above] {} (B01);

        \path (A20) edge [loop left] node[above] {} (A20);
        \path (B20) edge [loop left] node[above] {} (B20);

      \end{tikzpicture}      
      \caption{A consensus game acceptor}
  \label{fig:game-cover}
  \end{center}
\end{figure}

\section{Describing languages by games}

We consider languages $L$ over a terminal alphabet $\Sigma$. The empty word
$\varepsilon$ is excluded from the language, and also from its complement 
$\bar{L} := (\Sigma^* \setminus \{ \varepsilon \} ) \setminus L$. 
As acceptors for such languages, 
we consider games over 
an observation alphabet $\Gamma \supseteq \Sigma$, and we assume that 
no observation sequence in $\Sigma^+$ is omitted:
for every word $w \in \Sigma^+$, 
and each player $i$,
there exists a play~$\pi$ that yields the observation sequence
$\beta^i( \pi ) = w$.

Given a consensus game acceptor $G$, 
we associate to every observation-based strategy $s \in S^1$ of the first player, the
language $L( s ) := \{ w \in \Sigma^* ~|~ s( w ) = 1 \}$.
We say that the game $G$ \emph{covers} a language $L \subseteq
\Sigma^*$, if $G$ is solvable and
\begin{itemize}
\item $L = L( s )$, for \emph{some} winning strategy $s \in S^1$, and
\item $L \subseteq L ( s )$, for \emph{every} winning strategy $s \in S^1$.
\end{itemize}
If, moreover,  $L = L( s )$ for \emph{every} winning strategy in $G$, we say that  
$G$ \emph{characterises}~$L$. In this case, all winning strategies map $L$ to 1
and $\bar{L}$ to~$0$.

Notice that every solvable game covers a unique language~$L$ over the
full observation alphabet~$\Gamma$. With respect to a given terminal alphabet $\Sigma \subseteq
\Gamma$, the covered language is hence $L\cap \Sigma^*$. 
For instance, the consensus game acceptor represented in Figure~\ref{fig:game-cover}
covers the language $\{a^nb^n~|~n \in \mathbb{N}\}$ over $\{a, b \}$.  
To characterise a language rather than covering it, we need to add
constraints that require to reject inputs. 

Given two games $G, G'$,  
we define
the \emph{union} $G \cup G'$ as the consensus game obtained
by taking the disjoint union of $G$ and $G'$ 
and identifying the initial states.
Then, winning strategies of the component games can be turned into winning
strategies of the composite game, if they agree on 
the observation sequences over the common alphabet.

\begin{lemma}\label{lem:conjunction}
Let $G$, $G'$ be two consensus games over observation alphabets $\Gamma$, $\Gamma'$. 
Then, an observation-based strategy $r$ is winning in $G \cup G'$ if, and only if, 
there exist 
observation-based winning strategies $s,s'$ in $G$, $G'$ that agree 
with~$r$ on $\Gamma^*$ and on~$\Gamma'^*$, respectively.
\end{lemma}
\begin{proof}
Clearly, if $r$ is a winning strategy in $G \cup G'$, 
then its restrictions $s$, $s'$ to $\Gamma^*$ and
$\Gamma'^*$ are winning strategies in $G$ and $G'$.
For the converse, let $s, s'$ be two observation-based 
winning strategies in $G$, $G'$ that agree on $( \Gamma \cap \Gamma' )^*$. 
Notice that every observation sequence in the conjunction $G \cup G'$ is 
either included in $\Gamma^*$ or in $\Gamma'^*$, 
as a play that enters $G$ at the initial state can never
reach $G'$ and vice versa.
Therefore, the function $r: (\Gamma \cup \Gamma')^* \to \{ 0, 1\}$ 
that agrees with $s$ on 
$\Gamma^*$ and with $s'$ on $\Gamma'^*$ 
determines a winning strategy in the acceptor game $G \cup G'$.
\end{proof}

Whenever a language and its
complement are covered by two consensus games, 
we can construct a new game that characterises the language.
The construction involves \emph{inverting} the decisions
in a game, that is, replacing the admissible decisions
for every final state $v \in V$ with $\Omega( v ) = \{0\}$ 
by $\Omega( v ) := \{1\}$ and vice versa; 
final states $v$ with $\Omega( v ) = \{0, 1\}$ remain unchanged.

\begin{lemma}\label{lem:cover-char}
Suppose two consensus games $G$, $G'$ 
cover a language~$L \subseteq \Sigma^*$ and its complement~$\bar{L}$, respectively. 
Let $G''$ be the game obtained
from $G'$ by inverting the admissible decisions.
Then, the game $G \cup G''$ characterises $L$.
\end{lemma}

\begin{proof}
  Let $G$, $G'$ be two acceptor games 
  that cover $L$ and $\bar{L}$. Witout loss of generality, we assume that
  the observation alphabets $\Gamma$, $\Gamma'$ intersect only on 
  $\Gamma \cap \Gamma' = \Sigma$; other 
  common observations may be renamed.
  We need to show that the union $G \cup G''$ is solvable
  and that every winning strategy~$r$ of Player~$1$ satisfies $L = L( r )$.

  For the first point, let $s$,$s'$ be two observation-based 
  winning strategies for Player 1 
  in $G$,$G'$ such that $L( s ) = L$ and $L( s' ) = \bar{L}$. 
  Then, the strategy~$s''$ for~$G''$ with
  $s''( \pi ) := 1 - s'( \pi )$ for any play~$\pi$ is winning, 
  since $s'$ is winning in $G'$ and the admissible decisions are inverted.
  Moreover,~$s$ and~$s''$ agree on all sequences of observations in
  $\Sigma^*$, 
  so they admit a common extension $r$ to $G \cup G''$ 
  that determines a winning strategy, according to Lemma~\ref{lem:conjunction}. 

  For the second point, consider an
  arbitrary winning strategy $r$ in the composite game $G \cup G''$. 
  Then, the restrictions of $r$ to $\Gamma^*$ and $\Gamma'^*$ are
  winning in~$G$ 
  and~$G''$, respectively.
  This implies, for all $w \in \Sigma^*$, 
  that $r( w ) = 1$ if $w \in L$, and 
  $r(w) = 0$ otherwise: the former because $G$ covers~$L$, and the latter
  because $G'$ covers~$\bar{L}$ and the decisions in $G''$ are inverted.
  Hence, we have $L( r )=L$ for
  every winning strategy $r$ in the game $G \cup G''$, 
  which thus characterises $L$. 
\end{proof}

\subsection{Domino frontier languages}

We use domino systems as an alternative to 
encoding machine models and formal grammars~(See \cite{Boas97} for a survey.).  
A \emph{domino system} $\mathcal{D} = (D, E_h, E_v)$ 
is described by a finite set of \emph{dominoes} 
together with a horizontal and a vertical compatibility relation 
$E_h, E_v \subseteq D \times D$.
The generic domino tiling problem 
is to determine, for a given system~$\mathcal{D}$, 
whether copies of the dominoes can be arranged to
tile a given space in the discrete grid $\mathbb{Z} \times \mathbb{Z}$, 
such that any two vertically or horizontally 
adjacent dominoes are compatible. 
Here, we consider finite rectangular grids 
$Z(\ell, m) := \{ 0, \dots, \ell+1\} \times \{ 0, \dots, m \}$, 
where the first and last column, and the bottom row are distinguished as 
border areas. Then, the question is
whether there exists a \emph{tiling} 
$\tau : Z(\ell, m) \to D$ 
that assigns to every point $(x, y) \in Z(\ell, m)$ 
a domino $\tau( x, y ) \in D$
such that:
\begin{itemize}
\item if $\tau(x, y) = d$ and $\tau(x + 1, y) = d'$
then $(d, d') \in E_h$, and
\item if $\tau(x, y) = d$ and $\tau(x, y + 1) = d'$
then $(d, d') \in E_v$.
\end{itemize}

The \emph{Border-Constrained Corridor} tiling 
problem takes as input a domino system $\mathcal{D}$ with 
two distinguished border dominoes $\#$ and $\Box$, together
with a sequence $w = w_1\,w_2\, \dots\,w_\ell$ of dominoes $w_i \in D$, 
and asks whether there exists a height $m$ such that the 
rectangle $Z(\ell, m )$ 
allows a tiling~$\tau$ with $w$ in the top row, $\#$ in the first and last 
column, and $\Box$ in the bottom row:
\begin{itemize}
\item
  $\tau( i, 0 ) = w_i$, for all $i = 1, \dots, \ell$;
\item
  $\tau( 0, y ) = \tau( \ell + 1, y ) = \#$, for all $y = 0, \dots, m - 1$;
\item
  $\tau( x, m ) = \Box$, for all $x = 1, \dots, \ell$.
\end{itemize}

Domino systems can be used to recognise formal languages. 
For a domino system $\mathcal{D}$ with side and bottom 
border dominoes as above, the \emph{frontier 
language} $L( \mathcal{D} )$ is the set of words $w \in D^*$ that 
yield positive instances of the border-constrained corridor
tiling problem. 
We use the following correspondence between context-sensitive
languages and domino systems established by Latteux and Simplot.

\begin{theorem}[\cite{LatteuxSim97a,LatteuxSim97b}] \label{thm:cs-domino}
For every context-sensitive language $L \subseteq \Sigma^*$, we can
effectively construct a
domino system $\mathcal{D}$ over a set of dominoes $D \supseteq
\Sigma$ with frontier language~$L(\mathcal{D}) = L$.
\end{theorem}

Figure~\ref{fig:dominos} describes a domino system for recognising the
language $a^nb^n$ also covered by the game in Figure~\ref{fig:game-cover}.
In the following, we show that domino systems can generally be described in terms
of consensus game acceptors. 

\subsection{Uniform encoding of domino problems in games}
Game formulations of domino tiling problems
are standard in complexity theory, going back to the early work of Chlebus~\cite{Chlebus86}. 
However, these reductions are typically non-uniform: they construct
for every input instance consisting of a domino system together with a border
constraint a different game which depends, in particular, on the size of
the constraint. 
Here, we use imperfect information to define a \emph{uniform} reduction that associates to
a fixed domino system $\mathcal{D}$ a game $G(\mathcal{D})$,
such that for every border constraint $w$, the question whether
$\mathcal{D}, w$
allows a correct tiling is reduced to the question of whether 
decision~$1$ is safe in a certain play associated to $w$ in $G(\mathcal{D})$. 
 
\begin{figure}[t]
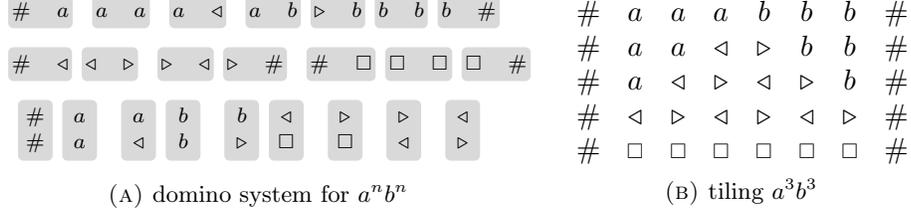

\begin{center}
\subfloat[domino system for $a^n b^n$]{
\label{fig:domino-system}
\tikzset{every picture/.style={node distance=1em,->,>=stealth',shorten >=1pt,auto}} 
\tikzstyle{every node}=[rounded corners=2pt, fill=black!15, 
draw=none, minimum width=.4em, inner sep=1.4pt, font=\scriptsize]
\begin{tabular}{l}
\tikz \node{$\begin{matrix}
\# & a \\
\end{matrix}
$};
\ 
\tikz \node{$\begin{matrix}
a & a \\
\end{matrix}
$};
\ 
\tikz \node{$\begin{matrix}
a & \triangleleft \\
\end{matrix}
$};
\ 
\tikz \node{$\begin{matrix}
a & b \\
\end{matrix}
$};
~\tikz \node{$\begin{matrix}
\triangleright & b \\
\end{matrix}
$};
~\tikz \node{$\begin{matrix}
b & b \\
\end{matrix}
$};
~\tikz \node{$\begin{matrix}
b & \# \\
\end{matrix}
$};
\\[.3em] 
\tikz \node{$\begin{matrix}
\# & \triangleleft \\
\end{matrix}
$};
~\tikz \node{$\begin{matrix}
\triangleleft & \triangleright \\
\end{matrix}
$};
\ 
\tikz \node{$\begin{matrix}
\triangleright & \triangleleft\\
\end{matrix}
$};
~\tikz \node{$\begin{matrix}
\triangleright & \# \\
\end{matrix}
$};
\ 
\tikz \node{$\begin{matrix}
\# & \mbox{\tiny $\Box$}\\
\end{matrix}
$};
~\tikz \node{$\begin{matrix}
\mbox{\tiny $\Box$}& \mbox{\tiny $\Box$} \\
\end{matrix}
$};
~\tikz \node{$\begin{matrix}
\mbox{\tiny $\Box$} & \# \\
\end{matrix}
$};\\[.3em]
\tikzstyle{every node}=[rounded corners=2pt, minimum height=2.1em, fill=black!15,
draw=none, minimum width = 1.2em,
inner sep=1pt, font=\scriptsize]
\tikz \node{$\begin{matrix}
\# \\
\#
\end{matrix}
$}; 
~\tikz \node{$\begin{matrix}
a \\
a
\end{matrix}
$}; \,\
\tikz \node{$\begin{matrix}
a \\
\triangleleft
\end{matrix}
$}; 
~\tikz \node{$\begin{matrix}
b  \\
b 
\end{matrix}
$}; \,\
\tikz \node{$\begin{matrix}
b \\
\triangleright
\end{matrix}
$};
~\tikz \node{$\begin{matrix}
\triangleleft \\
\mbox{\tiny $\Box$}
\end{matrix}
$}; \,\
\tikz \node{$\begin{matrix}
\triangleright \\
\mbox{\tiny $\Box$}
\end{matrix}
$}; \,\
\tikz \node{$\begin{matrix}
\triangleright \\
\triangleleft
\end{matrix}
$}; \,\
\tikz \node{$\begin{matrix}
\triangleleft\\
\triangleright
\end{matrix}
$};
\end{tabular}
}
\subfloat[tiling $a^3b^3$]{
  \label{fig:tiling}
  \begin{tabular}{cccccccc}
$\#$ & $a$ & $a$ & $a$ & $b$ & $b$ & $b$ & $\#$ \\
$\#$ & $a$ & $a$ & $\triangleleft$ & $\triangleright$ & $b$ & $b$ & $\#$  \\
$\#$ & $a$ & $\triangleleft$ & $\triangleright$ & $\triangleleft$ & $\triangleright$ & $b$ & $\#$ \\
$\#$ & $\triangleleft$ & $\triangleright$ & $\triangleleft$ & $\triangleright$ & $\triangleleft$ & $\triangleright$ & $\#$ \\
$\#$ & \scriptsize{$\Box$} & \scriptsize{$\Box$}&\scriptsize{$\Box$} &
\scriptsize{$\Box$} &\scriptsize{$\Box$} & \scriptsize{$\Box$} & $\#$ \\
\end{tabular}
}
\end{center}
\caption{Characterising a language with dominoes}
\label{fig:dominos}
\end{figure}

\begin{proposition}\label{prop:domino-games}
  For every domino system $D$, we can construct, in polynomial time, a
  consensus game acceptor 
  that covers the frontier language of~$D$.
\end{proposition}

\begin{proof}
Let us fix a domino system $\mathcal{D} = (D, E_h, E_v)$ with a left
border domino~$\#$ and a bottom domino~$\Box$.  
We construct a consensus game $G$ for the alphabet 
$\Sigma := D \setminus \{ \#, \Box \}$ to cover the frontier
language $L(\mathcal{D})$. There are
domino states of two types: singleton states $d$ for each $d \in
D \setminus \{ \# \}$ and pair states $(d, b)$ for each $(d, b) \in E_v$.
At singleton states $d$, the two players receive the
same observation $d$. 
At states $(d, b)$, the first player observes $d$ and the second player $b$.
The domino states are connected by transitions 
$d \to d'$ for every $(d, d') \in E_h$, and
$(d,b) \to (d',b')$ whenever $(d, d')$
and $(b,b')$ are in $E_h$. 
There is an initial state $v_0$
and two final states $\widehat{z}$ and $z$, all associated to the
the observation $\#$ for the border domino.
From~$v_0$ there are transitions to all 
compatible domino states $d$ with $(\#, d) \in E_h$, 
and all pair states $(d, b)$ with $(\#, d)$ and $(\#,
b) \in E_h$. 
Conversely, the final state~$z$ is reachable from all 
domino states $d$ with $(d, \#) \in E_h$, 
and all pair states $(d, b)$ with $(d, \#)$ and 
$( b, \#) \in E_h$; the final $\widehat{z}$ is reachable only from
the singleton bottom domino state~$\Box$.  
Finally, admissible decisions are $\Omega(z) = \{ 0, 1\}$ 
and $\Omega( \hat{z} ) =  \{ 1 \}$. 
Clearly, $G$ is a consensus game, 
and the construction can be done in polynomial time.

Note that any sequence $x = d_1\, d_2\, \dots\, d_\ell \in D^\ell$
that forms a horizontally consistent row in a tiling  by $\mathcal{D}$
corresponds in the game to a play 
$\pi_x = v_0\, d_1\, d_2\, \dots\, d_\ell, z$ 
or $\pi_x = v_0\, \Box^\ell\, \widehat{z}$.
Conversely, every play in $G$ 
corresponds either to one possible row, 
in case Input chooses a single domino 
in the first transition, or to two rows, in case it chooses a pair.
Moreover, a row $x$ can appear 
on top of a row $y = b_1\, b_2\, \dots\, b_\ell \in D^\ell$ in a tiling 
if, and only if, 
there exists a play~$\rho$  in $G$ such that 
$\pi_x \sim^1 \rho \sim^2 \pi_y$, namely 
$\rho = v_0\, (d_1, b_1)\, (d_2,b_2)\, \dots (d_\ell,
b_\ell)\, z$.

Now, we claim that at an observation sequence $\pi = w$  for 
$w \in \Sigma^\ell$ 
the decision~$0$ is safe
if, and only if, there exists no correct corridor tiling by
$\mathcal{D}$ with $w$ in the top row. 
According to our remark, 
there exists a correct tiling of the corridor with top row $w$, 
if and only if, there exists a sequence
of rows corresponding to plays $\pi_1, \dots, \pi_m$,
and a sequence of witnessing plays 
$\rho_1, \dots, \rho_{m-1}$ such that 
$w = \pi_{1} \sim^1 \rho_1 \sim^2 \pi_{2} \dots \sim^1
 \rho_{m-1} \sim^2 \pi_{m} = \Box^\ell. $
However, the decision~$0$ is unsafe  
in the play $\Box^\ell$ and therefore at $w$ as well. 
Hence, every winning
strategy $s$ for $G$ must prescribe $s(w) = 1$, 
for every word $w$ in the frontier language of $\mathcal{D}$, meaning
that $L( s ) \subseteq L( \mathcal{D} )$.

Finally, consider 
the mapping $s: D^* \to A$ that prescribes $s(w) = 1$
if, and only, if $w \in L( \mathcal{D} )$. The 
observation-based strategy $s$ in the consensus game~$G$ 
is winning since $s( \Box^*) = 1$, and it witnesses the condition
$L( s ) = L( \mathcal{D} )$. 
This concludes the proof that the constructed consensus game $G$ covers the frontier
language of~$\mathcal{D}$. 
\end{proof}

\section{Characterising context-sensitive languages}
Our first result establishes a
correspondence between context-sensitive languages and 
consensus games.

\begin{theorem}\label{thm:char-cs}
For every context-sensitive language $L \subseteq \Sigma^*$, 
we can construct effectively a consensus game acceptor 
that characterises $L$.
\end{theorem}

\begin{proof}
Let $L\subseteq \Sigma^*$ be an arbitrary context-sensitive language,
represented, e.g., by a linear-bounded automaton.
By Theorem~\ref{thm:cs-domino}, 
we can effectively construct a domino system~$\mathcal{D}$ with frontier
language~$L$. 
Further, by Proposition~\ref{prop:domino-games},
we can construct a consensus game $G$ that covers $L ( \mathcal{D} ) =
L$.
Due to the  Immerman-Szelepcs\'enyi Theorem, 
context-sensitive languages are effectively closed under complement, so we
can construct a consensus game $G'$ that covers $\Sigma^* \setminus L$ following 
the same procedure.
Finally, we combine the games $G$ and $G'$ as described in
Lemma~\ref{lem:cover-char} to obtain
a consensus game that characterises $L$.
\end{proof}

One interpretation of the characterisation is that, 
for every context-sensitive language, 
there exists a consensus game that is as hard to
play as it is to decide membership in the language.
On the one hand, this implies that winning strategies for consensus games are in
general $\PSPACE$-hard. Indeed, there are instances of
consensus games that admit winning strategies, 
however, any machine that computes the decision to take in a play
requires space polynomial in the length of the play.

\begin{theorem}
There exists a solvable consensus game for which 
every winning strategy is $\PSPACE$-hard. 
\end{theorem}
\begin{proof}
There exist context-sensitive languages with
a $\PSPACE$-hard word problem~\cite{Kuroda64}. 
Let us fix such a language $L \subseteq \Sigma^*$ together with a  
consensus game $G$ that characterises it, according to
Theorem~\ref{thm:char-cs}. 
This is a solvable game, and every winning strategy 
can be represented as an observation-based strategy~$s$ for the first player.
Then, the membership problem for $L$ reduces (in linear time) to the problem of
deciding the value of $s$ in a play in $G$: 
For every input word $w \in \Sigma^*$, we have $w \in L$ if, and only if, $s( w ) =
1$. In conclusion, for
every winning strategy $s$ in $G$, it is $\PSPACE$-hard to decide whether $s( \pi ) = 1$.
\end{proof}

On the other hand, it follows that 
determining whether a consensus game admits a winning strategy 
is no easier than solving the emptiness problem of
context-sensitive languages, which is well known to be undecidable.

\begin{theorem}\label{thm:undecidable}
The question whether a consensus game admits a winning strategy is
undecidable.  
\end{theorem}

\begin{proof}
We reduce the emptiness problem for a context-sensitive grammar
to the solvability problem for a consensus game.
For an arbitrary context-sensitive language $L\in \Sigma^*$ given as a linear bounded automaton, 
we construct a consensus game $G$ that characterises $L$, in polynomial time, 
according to Theorem~\ref{thm:char-cs}. 
Additionally, we construct a consensus game~$G'$
that characterises the empty language over $\Sigma^*$:
this can be done, for instance, by connecting a clique over letters in $\Sigma$
observable for both players to a final
state at which only the decision~$0$ is admissible.
Now, for any word $w \in \Sigma^*$, 
the game $G'$ requires decision $0$ at every observation sequences $w
\in \Sigma^*$, 
whereas $G$ requires decision $1$ whenever $w \in L$.
Accordingly, the consensus game $G \cup G'$ is solvable if, and only
if, $L$ is empty. As the emptiness problem for context-sensitive languages
is undecidable~\cite{Kuroda64}, it follows that the solvability
problem is undecidable for consensus game acceptors.
\end{proof}

We have seen that every context-sensitive language corresponds to a
consensus game acceptor such that language membership tests reduce to
winning strategy decisions in a play. 
Conversely, every solvable game admits a winning strategy that is the
characteristic function of some context-sensitive
language. Intuitively, a strategy should prescribe $0$ at a play 
$\pi$ whenever there exists a connected play
$\pi'$ at which~$0$ is the only admissible
decision. Whether this is the case can be verified by a
nondeterministic machine using space linear
in the length of the play $\pi$. 

\begin{theorem}
Every solvable consensus game 
admits a winning strategy that is 
implementable by a nondeterministic 
linear bounded automaton.
\end{theorem}

\section{Consensus and iterated transductions}

Our aim in the following is to investigate how the
structure of a consensus game relates to the complexity of the
described language which, in turn, 
determines the complexity of winning strategies.
Towards this, we view games as finite-state automata 
representing the relation between the observation sequences received by the 
players and the admissible decisions.

A synchronous \emph{transducer} 
is a two-tape automaton $(Q, \Gamma, \Delta, q_0, F)$ over an alphabet $\Gamma$,
with state set~$Q$, a transition relation 
$\Delta \subseteq Q \times \Gamma \times \Gamma \times Q$ labelled by
pairs of letters, an initial state $q_0 \in Q$ and a non-empty set $F \subseteq
Q$ of final states; in contrast to games, final states of transducers may have outgoing
transitions. We write $p \xrightarrow{a|b} q$ to denote a transition $(p, a,
b, q) \in \Delta$.
An \emph{accepting run} of the transducer is a path 
$\rho = q_0 \xrightarrow{a_1|b_1} q_1 \xrightarrow{a_2|b_2} \dots
\xrightarrow{a_n|b_n} q_n$ that follows transitions in~$\Delta$ starting from the
initial state $q_0$ and ending at a final state $q_n \in F$. 
The \emph{label} of the run is the pair of words 
$(a_1 \dots a_n, b_1, \dots b_n)$. A pair of words $(w, w') \in
\Gamma^* \times \Gamma^* $ is \emph{accepted}  by the transducer
 if it is the label of some accepting run.
The relation \emph{recognised} by the transducer is the 
set $R \subseteq \Gamma^* \times \Gamma^*$ of accepted pairs of words.
A relation is \emph{synchronous} if it is recognised by a synchronous
transducer. In general, we do not distinguish notationally between transducers and
the relation they recognise.
For background on synchronous, or
letter-to-letter, transducers, 
we refer to the survey~\cite{Berstel79} of Berstel
and to Chapter IV in the book~\cite{Sakarovitch2009} of Sakarovitch.

Given a consensus game acceptor~$G = (V, E, \beta^1, \beta^2, v_0, \Omega)$
over an observation alphabet~$\Gamma$, we define the \emph{seed}
of $G$ to be the triple $(R, L_\acc, L_\rej)$ 
consisting of the relation 
$
R := \{\, (\beta^1( \pi ),  \beta^2( \pi )) \in (\Gamma \times \Gamma )^*~|~ 
\pi~\text{play in}~G \, \}
$ 
together with the languages 
$L_\acc \subseteq \Gamma^*$ and $L_\rej \subseteq \Gamma^*$ of observation sequences 
$\beta^1( \pi )$ on plays~$\pi$ in $G$ 
with $\Omega( \pi ) = \{ 1 \}$ and $\Omega( \pi ) = \{ 0 \}$,
respectively.
The seed of any finite game can be represented by finite-state automata.

\begin{lemma}
For every consensus game, the seed languages $L_\acc, L_\rej$ are regular
and the seed relation $R$ is recognised by a synchronous transducer.  
\end{lemma}

\begin{proof} 
We construct automata from the game graph by moving 
observations from each state to the incoming transitions. 
Let $G = (V, E, \beta, v_0, \Omega)$ be a consensus game over an observation
alphabet~$\Gamma$, and let $(R, L_\acc, L_\rej)$ be its seed. 
We define three automata over the alphabet~$\Gamma$, on the subset
of~$V$ consisting of non-final game states and with initial state $q_0 = v_0$.

The automata for the seed languages~$L_\acc$ and $L_\rej$ allow
transitions $u \xrightarrow{a} v$ if $(u, v) \in E$ and $\beta^1( v ) = a$.
The set of final states 
consists of all game states~$v$ with an outgoing transition $(v,v') \in E$ 
to some final state with $\Omega( v' ) = \{ 1 \}$ for~$L_\acc$, and
with $\Omega( v' ) = \{ 0 \}$ for $L_\rej$.
Then, for all words $a_1 \dots a_n \in \Gamma^*$, accepting runs 
$v_0 \xrightarrow{a_1} v_1 \xrightarrow{a_2} \dots \xrightarrow{a_n} v_n$ of the automata
$L_\acc$ and $L_\rej$ correspond to plays $\pi := v_0 v_1 \dots v_n v'$ 
with observations $\beta^1( \pi ) = a_1 \dots a_n$ such that 
$\Omega( \pi ) = \{ 1 \}$ and $\Omega (\pi ) = 0$,
respectively. Hence, the automata recognise $L_\acc$ and $L_\rej$.

Similarly, the transducer for the seed relation has transitions
$u \xrightarrow{a|b} v$ whenever $(u, v) \in E$ with
$\beta^1( v ) = a$, $\beta^2( v ) = b$, and its final states are
states $v\in V$ with an outgoing transition $(v, v') \in E$ to some terminal state $v'$ in $G$.
Then, for any pair of words $( a_1 \dots a_n,  b_1 \dots b_n)  \in \Gamma^* \times\Gamma^*$,  
there exists an accepting transducer run
$
v_0 \xrightarrow{a_1|b_1} v_1 \xrightarrow{a_2|b_2} \dots
\xrightarrow{a_n|b_n} v_n
$
if, and only if, there exists
a play $\pi = v_0, v_1, \dots, v_n, v'$ with
$\beta^1( \pi ) = a_1 \dots a_n$, $\beta^2( \pi' ) = b_1 \dots b_n$,
for some final game state $v'$. 
So, the transducer recognises the seed relation $R$, as intended.
\end{proof}

Conversely, we can turn synchronous transducers 
and automata over matching alphabets into games.

\begin{lemma}
Given a synchronous relation $R \subseteq \Gamma^* \times \Gamma^*$ 
and two disjoint regular languages 
$L_\acc$, $L_\rej \subseteq \Gamma^*$, 
we can construct a consensus game with seed $(R, L_\acc, L_\rej)$. 
\end{lemma}

\begin{proof}
Let us consider a synchronous transducer $\calR$ and
two word automata $\calA$, $\calB$
that recognise $R$, $L_\acc$
and $L_\rej$, respectively. We assume that the word automata are
deterministic, with transition functions
$\delta^\calA: Q^\calA \times \Gamma \to Q^\calA$ 
and $\delta^\calB: Q^\calB \times \Gamma \to Q^\calB$.
To avoid confusion, we label the components of each automaton with its name 
and write $p \xrightarrow[\calR]{a} q$ for the transitions in
$\calR$. 

We construct a game~$G$ 
over the
observation alphabet $\Gamma \cup \{\# \}$ with states
formed of three components: a transducer transition $p
\xrightarrow[\calR]{a|b} q$, 
a state of $\calA$, and one of $\calB$. The game transitions follow 
the adjacency graph of the transducer in the first component 
and update the state of~$\calA$ and $\calB$ in the second and
third component according to the observation~$a$ of player~$1$ in the
first component.
More precisely, the set of game states is
 $V := \Delta \times Q^\calA \times Q^\calB \cup \{ v_0,
 v_{\acc}, v_{\rej}, v_{=} \}$, where $v_0$ is a fresh initial state 
whereas $v_\acc$, $v_\rej$ and $v_{=}$ are final
states.
Game transitions lead from the
initial state $v_0$ to all states 
$(q_0^\calR \xrightarrow[\calR]{a|b} q, q_0^\calA, q_0^\calB)$;
for each pair of incident transducer 
transitions $e := p \xrightarrow[\calR]{a|b} q$ and
$e' := q \xrightarrow[\calR]{a'|b'} q'$, and for all 
automata states $q^\calA \in Q^\calA$, $q^\calB \in Q^\calB$,  
there is a game transition 
from $(e, q^\calA, q^\calB)$ to $(e', \delta^\calA ( q^\calA, a ), \delta^\calA ( q^\calB, a
))$; finally, from any state 
$(p \xrightarrow[\calR]{a|b} q, q^\calA, q^\calB )$ with $q \in
F^\calR$ there is a transition to $v_\acc$ if 
$q^\calA \in F^\calA$ , to $v_\rej$ if $q^\calB \in F^\calB$, and
otherwise to $v_=$.
The observation at state 
$( p \xrightarrow[\calR]{a|b } q, q^\calA, q^\calB)$ 
is $(a, b)$; at the initial and the final states both players observe $\#$.
Admissible decisions are $\Omega( v_\acc ) = \{ 1 \}$, 
$\Omega( v_\rej ) = \{ 0 \}$, and 
$\Omega( v_= ) = \{ 0,1 \}$.

Now, for any play $\pi$, the first component corresponds to an accepting
run of~$\calR$ on the pair of observation sequences $(\beta^1(\pi),
\beta^2( \pi ))$, whereas the second and third components correspond to
runs of $\calA$ and $\calB$ on $\beta^1( \pi )$ which are accepting
if $\Omega( \pi ) = \{ 1 \}$ and $\Omega( \pi ) = \{ 0 \}$,
respectively.
Accordingly, the game $G$ has seed $(R, L_\acc, L_\rej)$. 
\end{proof}


Thanks to the translation between games and automata, we can reason
about games in terms of 
elementary operations on their seed. 
Our notation is close to the one of
Terlutte and Simplot \cite{TerlutteSim00}.
Given a relation $R \subseteq \Gamma^* \times \Gamma^*$, the inverse
relation is $R^{-1} := \{ (x, y) \in \Gamma^* \times \Gamma^* ~|~ (y, x) \in R
\}$. The composition of $R$ with a relation
$R' \in \Gamma^*$ is 
$ 
R  R' := \{ (x, y) \in (\Gamma \times \Gamma)^*~|~(x,z) \in R \text{ and }
(z,y) \in R' \text{ for some } z \in \Gamma^* \}.
$
For a language $L \subseteq \Gamma^*$, we write 
$R L := \{x \in \Gamma^* ~|~ (x,y) \in R~\text{and}~y \in L \}$. 
For a subalphabet $\Sigma \subseteq \Gamma$, we 
denote the identity relation by $(\cap\Sigma^*) := \{ (x, x) \in \Sigma^* \times \Sigma^*
\}$. 
The power $R^k$ of $R$ is defined by $R^0 :=
(\cap\Gamma)$, and $R^{k+1} := R^k R$ for $k > 0$. Finally, the  
iteration of $R$ is $\displaystyle{R^* := \cup_{0 \le k < \omega} R^k}$.

One significant relation obtained from the seed transducer~$R$ of a game~$G$ is the
\emph{reflection} relation $\tau( R ) := {RR}^{-1}$. 
That is, a word $w \in \Gamma^*$ over the observation alphabet is a reflection
of $u \in \Gamma^*$ if, whenever player~$1$ observes $u$, player~$2$ considers it possible
that $1$ actually observes $w$. Obviously, this relation is reflexive
and symmetric. Its transitive closure relates observations received on
connected plays.
 
\begin{lemma}\label{lem:connected-iterated} 
Let $G$ be a consensus game with seed relation~$R \subseteq \Gamma^* \times
\Gamma^*$, and let $\tau := RR^{-1}$ be its reflection. Then,
\begin{enumerate}[(i)]
\item $\tau = \{\, (\, \beta^1( \pi
  ), \beta^1( \pi')) ~|~ \text{plays}~\pi \sim^2 \pi'~\text{in } G \,\}$, and
\item
$\tau^* = (\cap \Gamma^*) \cup \{\, (\beta^1( \pi
  ), \beta^1( \pi')) ~|~ \text{plays}~\pi \sim^* \pi'~\text{in } G \,\}$.
\end{enumerate}
\end{lemma}

\begin{proof}
$(i)$ By our definition of the seed relation,
for every pair of words $(w, w') \in RR^{-1}$, there exist plays $\pi, \pi'$ 
such that 
$\beta^1( \pi ) = w $, $\beta^1( \pi' ) = w'$, and $\beta^2( \pi ) =
\beta^2( \pi' )$, that is, $\pi \sim^2 \pi'$. 
Conversely, for any pair of plays $\pi \sim^2 \pi'$, the observation
sequences are related by 
$(\beta^1( \pi ), \beta^2( \pi ))\in R$
and $(\beta^2( \pi ), \beta^1( \pi' )) = (\beta^2( \pi' ), \beta^1(
\pi' )) \in R^{-1}$. Hence $(\beta^1( \pi ), \beta^1( \pi') )\in
RR^{-1}$. 

($ii$ ``$\supseteq$''): Clearly, $(\cap\Gamma^*) \subseteq \tau^*$.
Further, by definition of connectedness, if $\pi \sim^* \pi'$, then there exists a sequence of  
plays $(\pi_\ell)_{\ell \le 2k}$ with $\pi_0 = \pi$, 
$\pi_{2k} = \pi'$, and $\pi_{\ell} \sim^1
\pi_{\ell+1} \sim^2 \pi_{\ell+2}$ for all even $\ell < 2k$.
Therefore, the observation sequences 
\begin{align*} 
x := \beta^1( \pi_{\ell} ) = \beta^1(\pi_{\ell + 1}), \quad 
y := \beta^1( \pi_{\ell + 2} ), \quad \text{and} \quad z := \beta^2( \pi_{\ell + 1} ) =
\beta^2( \pi_{\ell + 2})
\end{align*} are
related by $(x, z) \in R$ and 
$(z, y) \in R^{-1}$, which means that  
$(x, y ) = (\beta^1( \pi_{\ell} ), \beta^1 ( \pi_{\ell+2} ) ) \in R R^{-1}$, for all
even $\ell < 2k$. Accordingly, we obtain
$(\beta^1( \pi ), \beta^1(\pi' )) \in (R R^{-1})^*$. 

($ii$ ``$\subseteq$'') To show that every pair of distinct words in $\tau^*$ can be observed
by player~$1$ on connected plays, 
 we verify by induction on the power $k \ge 1$ that
for every pair $(w, w') \in (RR^{-1})^k$, 
there exists a sequence of plays 
$\pi_0 \sim^1 \pi_1 \sim^2  \dots \sim^2 \pi_{2k}$ 
such that $\beta^1( \pi_0 ) = w$ and 
$\beta^1( \pi_{2k} ) = w'$.

The base case for $k = 1$ follows from point~$(i)$ of the present
lemma: if $(w, w') \in (RR^{-1})$, then there exist $\pi \sim^2 \pi'$ with $\beta^1( \pi ) = w$,
$\beta^1( \pi' ) = w'$, and we set $\pi_0 = \pi_1 = \pi$ and $\pi_2 = \pi'$.
For the induction step, assume that the hypothesis holds for a power $k \ge
1$ and consider
$(w, w') \in (RR^{-1})^{k+1}$. That is, there exists a word $z \in \Gamma^*$ 
such that $(w, z) \in (RR^{-1})^k$ and
$(z, w') \in (RR^{-1})$. The former implies, by induction hypothesis, that we
have a chain $\pi_0 \sim^1 \pi_1 \sim^2 \dots \sim^1 \pi_{2k-1} \sim^2 \pi_{2k}$ 
with $\beta^1( \pi_0 ) = w$ and $\beta^1( \pi_{2k} ) = z$; from the
latter it follows, by definition of $R$, that there exist plays
$\pi$ and $\pi'$ with 
$\beta^1( \pi ) = z = \beta^1( \pi_{2k} )$, 
$\beta^2( \pi ) = \beta^2( \pi' )$, and
$\beta^1( \pi' ) = w'$.
Therefore, we can prolong the witnessing chain by 
setting $\pi_{2k+1} := \pi$, $\pi_{2k+2} := \pi'$, which concludes the
induction argument.
\end{proof}

The consensus condition requires decisions to be invariant under 
the reflection relation. 
This yields the following characterisation of
winning strategies.

\begin{lemma}\label{lem:char-winning}
Let $G$ be a consensus game with seed $(R, L_\acc, L_\rej)$ over an alphabet 
$\Gamma$, and let $\tau := RR^{-1}$ be its reflection relation.
Then, a strategy~$s:\Gamma^* \to \{0, 1\}$ of player~$1$ is winning if,
and only if, it
assigns $s( w ) = 1$ to every observation sequence $w \in \tau^*L_\acc$
and $s( w ) = 0$ to every observation sequence $w \in \tau^*L_\rej$.
\end{lemma}

\begin{proof}
(``\!$\implies$\!'') According to Lemma~\ref{lem:connected-iterated}$(ii)$, every word $w \in \tau^*L_\acc$
corresponds to the observation sequence $\beta^1( \pi ) = w$ of a
play~$\pi$ in $G$, and there exists a connected play $\pi' \sim^* \pi$ with $\pi'
\in L_\acc$. Therefore, any winning strategy 
$s: \Gamma^* \to \{ 0, 1 \}$ for player~$1$ must assign $s( w ) = 1$ to
all $w \in L_\acc$ 
and further, to all $w \in \tau^*L_\acc$, by the conditions of
consensus and indistinguishability.  
Likewise, it follows that $s( w ) = 0$ for all $w \in \tau^*L_\rej$. 

(``\kern-.4em$\impliedby$\kern-.45em'')
Consider the mapping $s: V^* \to \{0, 1\}$ with
$s( \pi ) = 1$ if, and only if, $\beta^1( \pi ) \in \tau^* L_{\acc}$.
Then~$s$ is a valid strategy, as for all $\pi \sim^2 \pi'$ we have 
$(\beta^1( \pi ), \beta^1( \pi' )) \in \tau$ and $(\beta^1( \pi'), \beta^1( \pi )) \in \tau$, hence
 $\beta^1( \pi ) \in \tau^* L_{\acc}$ if, and only if, $\beta^1(\pi') \in \tau^*
 L_{\acc}$. 
If it is the case that $s(\pi) = 0$  for
all $\pi \in \tau^*L_\rej$, that is, $\tau^* L_\acc \cap \tau^* L_\rej =
\emptyset$, 
then $s$ is a winning strategy.
\end{proof}

As a direct consequence, we can characterise the language
defined by a game in terms of iterated transductions.

\begin{theorem}\label{thm:char-iterated}
Let $G$ be a consensus game with seed $(R, L_\acc, L_\rej)$, and let $\Sigma$ be
a subset of its alphabet. Then, for the reflection $\tau := RR^{-1}$,
we have: 
\begin{enumerate}[(i)]
\item $G$ is solvable if, and only if, $\tau^* L_\acc \cap \tau^* L_\rej = \emptyset$. 
\item If $G$ is solvable, then it covers the language $ (\cap \Sigma^*) \tau^*
  L_\acc$.
\item If $G$ is solvable and $(\cap \Sigma^*) (\tau^* L_\acc \cup \tau^*
  L_\rej) = \Sigma^*$, then~$G$ characterises the language $ (\cap \Sigma^*) \tau^*
  L_\acc$.
\end{enumerate}
\end{theorem}

\section{Games for context-free languages}

Properties of iterated letter-to-letter transductions, or
equivalently, length-preserving transductions, have been investigated
by Latteux, Simplot, and Terlutte
in~\cite{LatteuxEtAl1998,TerlutteSim00}, where it is also shown that iterated synchronous transductions 
capture context-sensitive languages. Our setting is, however, more restrictive in that
games correspond to symmetric transductions.
In the following, we investigate a family of consensus game acceptors
that captures context-free languages.
Since the class is not closed under complement, 
we will work with the weaker notion of
covering a language rather
than characterising it. For the language-theoretic discussion, 
we generally assume 
that the rejecting seed language $L_\rej$ is empty
and specify the seed $(R, L_\acc, \emptyset)$ as $(R, L_\acc)$.


Firstly, we remark that regular languages correspond to games where
the two players have the same
observation function.
Clearly, such games admit regular winning strategies whenever they are solvable.

\begin{proposition}
A language~$L \subseteq \Sigma^*$ is regular if, and only if, it is
characterised by a consensus game acceptor where the seed relation is the identity.
\end{proposition}
\begin{proof}
Every regular language~$L \subseteq \Sigma^*$ is
characterised by the game with seed $(\,(\cap \Sigma^*), L, \Sigma^*
\setminus L\,)$. 
Conversely, suppose a language~$L \subseteq \Sigma^*$ is
characterised by a game~$G$ over an alphabet~$\Gamma
\supseteq \Sigma$ with seed $(\, (\cap \Gamma^*), L_\acc, L_\rej \,)$. 
Then, $G$ also covers~$L$ and, 
by Theorem~\ref{thm:char-iterated}$(ii)$, it follows that  $L = (\cap \Sigma^*) 
(\cap \Gamma^*)^*L_\acc = \Sigma^* \cap L_\acc$. 
Hence, $L$ is regular. 
\end{proof}

\subsection{Dyck languages.}\label{para:Dyck}
As a next exercise, let us construct games for covering Dyck languages,
that is, languages of well-balanced words of brackets; 
we also allow neutral symbols, which may appear at any position
without affecting the bracket balance.
Our terminal alphabet~$\Lambda$ consists of an alphabet $B_n = \{~ [_k, ]_k~|~ 1 \le
k \le n ~\}$ of $n \ge 1$ matching brackets and a set $C$~of neutral symbols. 
For a word $w \in \Lambda^*$ and an index $k$, we denote by $\mathrm{excess}^k( w )$ 
the difference between the number of opening and of closing brackets
$[_k$ and $]_k$. 
Then, the Dyck language~$\DA$ over~$\Lambda$ consists of the words $w
\in \Lambda^*$ such
that, for each kind of brackets $k \in \{ 1,
\dots, k \}$, $\mathrm{excess}^k( w ) = 0$, whereas
for all
prefixes~$w'$ of~$w$,
$\mathrm{excess}^k( w' ) \ge 0$. 

Given a terminal alphabet $\Lambda = B_n \cup C$, we define the
transducer $R_{n,C}$ over the observation alphabet
$\Gamma := A \cup \{ \Box \}$ with a set of states
$\{q_0, q_1, \dots, q_n \}$ among which $q_0$ is the initial and the only final
state, and with the following two kinds of transitions:
\emph{copying} transitions 
$q_0 \xrightarrow{a|a} q_0$ for all $a \in \Gamma$, 
as well as $q_k \xrightarrow{\Box|\Box} q_k$ for 
every $k \in \{1, \dots n \}$, and 
\emph{erasing} transitions $q_0 \xrightarrow{[_k|\Box} q_k$ and 
$q_k \xrightarrow{]_k|\Box} q_0$ for the brackets of each kind $k$, and 
$q_0 \xrightarrow{c|\Box} q_0$ for each neutral symbol $c \in C$. 
Essentially, $R_{n, C}$ erases neutral symbols and 
any innermost pair of brackets.

\begin{lemma}\label{lem:Dyck}
The Dyck language over an alphabet~$n$ of matching brackets 
and a set~$C$ of neutral symbols is covered by the game with seed
$(R_{n,C}, \Box^*)$. 
\end{lemma}

\begin{proof}
  For a terminal alphabet~$\Lambda$ with partitions $B_n, C$ as in the statement, we denote
  the corresponding Dyck language by $\DA$ and
  the previously defined transduction by $R:= R_{n,C}$. 
  The observation alphabet $\Gamma = B_n \cup C \cup
  \{ \Box \}$ extends~$A$ with an additional neutral
  symbol~$\Box$; let $\DGamma \supseteq \DA$ be the Dyck language over
  this extended alphabet. 
  Consider now the game $G$ over~$\Gamma$ with seed $(R, \Box^*)$.  
 We using the reflection relation $\tau := R
 R^{-1}$ to argue that $\DA = \tau^* \Box^*$.

  To see that the Dyck language $\DA$ 
  is contained in the language $\tau^*\Box^*$ covered
  by $G$, observe that for every pair of words $u, u' \in \Gamma^*$ where
  $u \in \DGamma$ and $u'$ is obtained from $u$ by replacing one 
  innermost pair of matching brackets with $\Box$, we
  have $(u, u') \in R$. Since the relation 
  $R$ contains the identity on $\Gamma$, it follows that $(u, u')
  \subseteq R R^{-1}$, so $(u,
  u') \in \tau$. If we set out with an arbitrary word~$w \in \Lambda^*$,
 first  erase all neutral symbols 
by applying $R$ once, 
and then repeat applying~$R$ to erase an innermost pair of brackets, 
we end up with $\Box^*$, hence
  $w \in \tau^* \Box^*$.
 
  Conversely, to verify that every word in $\tau^* \Box^*$ has
  well-balanced brackets, we show that $\DGamma$ is invariant under
  the transductions $R$ and $R^{-1}$ in the sense that for any pair $(u, w) \in
  R \cup R^{-1}$, we have $u \in \DGamma$ if, and only if, $w \in \DGamma$. 
  Towards this, let us fix an accepting run of $R$ on $u|w$ and
  compare the values $\mathrm{excess}^k( u' )$ and $\mathrm{excess}^k(
  w' )$ of its prefixes $u'|w'$, for any~$k \ge n$: the values are
  equal until a transition
  $q_0 \xrightarrow{[_k|\Box} q_k$ is taken at some prefix
  $u'[_k|w'\Box$. Since $q_k$ is
      not final, the run will take the transition $q_k \xrightarrow{]_k|\Box}
q_0$ at some later position; let $u'']_k|w''\Box$ be the shortest continuation 
of $u'|w'$ at which $q_0$ is reached again.
Hence, we set out with
$\mathrm{excess}^k( w' )= \mathrm{excess}^k(u')$, 
and also have $\mathrm{excess}^k( u' ) = \mathrm{excess}^k( u'')$, since none of $[_k$ or
$]_k$ is transduced while looping in $q_k$; after returning to $q_0$, again
$\mathrm{excess}^k( w'')= \mathrm{excess}^k( u'')$, because the
opening bracket was matched. So, it is the case that 
$\mathrm{excess}^k( u' ) \ge 0$ for all prefixes~$u'$ of~$u$ 
if, and only if, $\mathrm{excess}^k( w') \ge 0$ for all prefixes~$w'$ 
of~$w$, which means that
$R \DGamma \subseteq \DGamma$ and $R^{-1} \DGamma \subseteq \DGamma$; 
the converse inclusions hold because $R$
contains the identity on $\Gamma$.  
Accordingly, for every sequence $w_0, \dots, w_\ell$ of words
  with $w_0 = w$ such that $(w_i, w_{i+1}) \in
  \tau$ for each $i < \ell$, we have $w_i
  \in \DGamma$ if, and only if, $w_{i+1} \in \DGamma$. 
  Since $\Box^* \in \DGamma$, it follows that $w \in \DGamma$, 
  for any word $w \in \tau^* \Box^*$.
  In conclusion,
  $(\cap \Lambda^*)\tau^* \Box^* \subseteq \DA$.
\end{proof}

\subsection{Context-free languages}
To extend the game description of Dyck languages to
arbitrary context-free languages, we use 
the Chomsky-Sch\"utzenberger representation
theorem~\cite{ChomskySch63} in the non-erasing 
variant of proved by Okhotin~\cite{Okhotin12}. 
A letter-to-letter homomorphism $h: \Lambda^* \to \Sigma^*$ is a
functional synchronous transduction that preserves concatenation, that
is, $h(uw) = h(u) h(w)$ for all words $u, w \in \Lambda^*$. 
Such a homomorphism is identified by its restriction $f: \Lambda \to \Sigma$
to single letters. 


\begin{theorem}[\cite{Okhotin12}]\label{thm:Okhotin}
A language $L \subseteq \Sigma^*$ is context-free if, and only if,
there exists a Dyck language $\DA$ over an alphabet $\Lambda$ of brackets
and neutral symbols, a regular language $M \subseteq \Lambda^*$, and a
letter-to-letter homomorphism $h: \Lambda^* \to \Sigma^*$, such that $L =
h(\, \DA \cap M \,)$. 
\end{theorem} 

We will show how a game that covers an arbitrary 
language~$L \subseteq \Lambda^*$ can be extended to cover
an homomorphic image of the intersection of~$L$ with a regular
language.
Let $h: \Lambda \to \Sigma$ be a letter-to-letter homomorphism, and
let $R$ be a synchronous transducer over an alphabet~$\Gamma \supseteq
\Lambda$.
We construct from~$R$ a new transducer $R_h$ 
over the enlarged alphabet $\Sigma \cup \Gamma \times \Lambda$ by
adding a \emph{coding} cycle. This is done by including a   
fresh final state~$q_h$, as well as transitions $q_0 \xrightarrow{h(a)|(a, a)} q_h$
and $q_h \xrightarrow{h(a)|(a, a)} q_h$ for all $a \in \Lambda$, and then
relabelling each transition $p \xrightarrow{a|b} q$ of $R$ to 
$p \xrightarrow{(a,x)|(b,x)} q$, for all $x \in \Lambda$.
Intuitively, the new transducer
duplicates the automaton tapes into two tracks which are both
initialised with a  
homomorphic pre-image $u \in \Lambda^*$ of a terminal word $w \in
\Sigma^*$, in a transduction via the coding cycle. The first track is intended to simulate~$R$ on the
pre-image~$u$, whereas the second track stores~$u$: 
the contents is looped through every other transduction
of $R_h$ or of its inverse. Notice that every run of~$R_h$
proceeds either through the new coding cycle, or through the original transducer
$R$, in the sense that, for any pair $(w, w')
\in R_h$ we have $\{ w, w' \} \subseteq \Sigma^* \cup
(\Gamma \times \Lambda)^*$.

\begin{lemma}\label{lem:hom-reg}
Suppose that a game acceptor with seed $(R, L_\acc)$
covers a language $L \subseteq \Lambda^*$. Let 
$M \subseteq \Lambda^*$ be a regular language and let
$h: \Lambda^* \to \Sigma^*$ be a letter-to-letter
homomorphism. 
Then, the game acceptor with seed $(R_h, L_\acc \times M)$ covers the
language $h( L \cap M )$ over the terminal alphabet~$\Sigma$.
\end{lemma}

\begin{proof}
Let~$G$ be a game over an alphabet~$\Gamma \supseteq \Lambda$ with
seed $(R, L_\acc)$.
Without loss of generality, we assume that~$R$
contains the identity on $\Lambda$, otherwise we take the reflexive 
transduction $RR^{-1}$ to obtain the seed of a game that covers the
same language. Further, let 
$M \subseteq \Lambda^*$~be a regular language  and 
let $h: \Lambda \to \Sigma$ represent a letter-to-letter homomorphism as in the statement.
We argue for the case where the alphabets 
$\Sigma$ and $\Gamma$ are disjoint; the general case follows by
composition with a relabelling homomorphism. 
Now, consider the game~$G'$ with seed transducer $R_h$ and
accepting language $L_\acc' := L_\acc \times M$. 
We denote the reflection relations associated $R$ and $R_h$
by $\tau:=RR^{-1}$ and $\tau' := R_h R_h^{-1}$.

To see that $h( L \cap M)$ is included in the 
language 
covered by~$G'$,
consider a word $w = h( u )$ for some $u \in \tau^* L_\acc \cap M$. 
By Theorem~\ref{thm:char-iterated}, 
there exists a witnessing sequence $(u_i)_{i \le \ell}$ with $u_0
= u$, $u_\ell \in L_\acc$, and $(u_i, u_{i+1}) \in \tau$ for
all~$i < \ell$. By construction of~$R_h$, we have $( w, (u, u)) \in
R_h$ which implies $( w, (u, u)) \in \tau'$,
thanks to our assumption that $(\cap \Lambda^*) \subseteq R$.
Since $(u_\ell, u) \in L_\acc \times M$, 
the sequence starting with $w$ 
and followed by $((u_i,u))_{i < \ell}$ is witnessing that $w \in \tau'^* L_\acc'$.

Conversely, consider a word $w \in (\cap \Sigma^*)\tau'^*L_\acc'$ 
and let $(w_i)_{i \le \ell}$
be a witnessing sequence with $w_0 = w$, $w_\ell \in L_\acc'$, and
$(w_i, w_{i+1}) \in \tau'$ for all $i \le \ell$. 
By construction of~$R_h$, the initial word $w$ 
is preserved at each term $w_i$ of the sequence, in the sense that
either $w_i = w$, or $w_i = (u_i, x_i ) \in \Gamma^* \times \Lambda^*$
for some $x_i \in \Lambda^*$
such that $w = h( x_i )$. 
By our assumption
that $\Sigma$ and $\Gamma$ are disjoint, we have $w \not \in
L_\acc$, so there exists a last position~$k < \ell$ 
with $w_k = w$. 
As $w \in \Sigma^*$ can only be transduced via the coding cycle,
it follows that $w_{i+1} = (u, u)$ for some $u \in \Lambda^*$ with
$h( u ) = w$. 
For each following position $i > k$, the terms of the sequence are of the form
$w_{i} = (u_i, u)$ for a certain word $u_i \in \Gamma^*$. Hence the coding
cycle cannot be applied and 
the sequence $(u_i)_{k < i \le \ell}$ satisfies,
$(u_i, u_{i+1}) \in \tau$ for all $i < \ell$. Moreover, 
$w_\ell = (u_\ell, u) \in L_\acc' = L_\acc \times M$. Thus, the
sequence witnesses that $u = u_{k+1} \in \tau^*L_\acc \cap M$, and since $h( u ) = w$, 
it follows that $w \in h( L \cap M)$. 
\end{proof}


Now, we can construct a game acceptor for covering an
arbitrary context-free language $L \subseteq \Sigma^*$
represented as $L = h( \DA \cap M)$ according to 
Theorem~\ref{thm:Okhotin}, by
applying Lemma~\ref{lem:hom-reg} to the particular case
of Dyck languages: 
we set out with the seed transducer~$R_{n,C}$ for the Dyck
language~$\DA$ over the
alphabet $\Lambda = B_n \cup C$
and add a coding cycle for the homomorphism $h$. This yields a
transducer~$R_h$ over the alphabet $\Sigma \cup (\Lambda \cup
\{\Box\}) \times \Lambda$ such that the game with seed $(R_h, \Box^* \times M)$
covers~$L$.

The generic construction of~$R_h$ can be
simplified in the case where $R = R_{n,C}$ is the seed transducer
of a Dyck language. Notice that, if we start from a word $w \in \Sigma^*$, 
then every distinct word $w' \neq w$ reached in the iteration
$(w, w') \in (R_h R_h^{-1})^*$ 
consists only of letters of the form
$(x, x)$ or $(\Box, x)$ with $x \in B_n \cup C$.
Hence, the \emph{reduced} transducer~$\hat{R}_h$ obtained from~$R_h$, by restricting to the 
(subset of transitions labelled with letters in the)
subalphabet $\Sigma \cup \{ (x, x)~|~x\in \Lambda \} \cup \{(\Box, x)~|~ x
\in \Lambda \} $
and identifying each pair $(x, x) \in A \times A$ with $x$, is equivalent to
$R_h$ in the sense that, for every regular language $M \subseteq
\Lambda^*$, the game with seed $(\hat{R}_h, \Box^* \times M)$
covers
the same language over~$\Sigma$ as the one with seed 
$(R_h, \Box^* \times M)$.
In contrast to~$R_h$, however, the reduced transducer~$\hat{R}_h$ has
fewer transition and a smaller alphabet, which extends the one of 
the underlying Dyck language~$\DA$ only with a 
\emph{neutralised} copy of each letter in $\Lambda$.

\newcommand*{\BoxedUp}[1]{\setlength{\fboxsep}{2.2pt}\Boxed{#1}}
\setlength{\fboxsep}{1.1pt}

We argue that the shape of the seed constructed above is prototypical for games that cover
context-free languages.
Therefore, we focus on games with a seed isomorphic 
to the seed $(\hat{R}_h, \Box^* \times M)$ obtained for the
homomorphic image of a Dyck-language over~$n$ bracket pairs
intersected with a regular language.
An \emph{$n$-flower} transducer is
a transducer~$R = (Q, \Gamma, \Delta, q_0, F)$ on a set of states
$Q = \{q_0, q_1, \dots, q_n, q_h \}$ with initial state~$q_0$
and final state set $F = \{q_0, q_h\}$, over an
alphabet that can be partitioned into $\Gamma = \Sigma \cup B_n \cup C
\cup \Lambda'$ where $B_n$ is a set of~$n$ matching
brackets~$[_k$, $]_k$ and $\Lambda'$ is a disjoint copy of $A: = B_n \cup C$
associating a neutralised variant $\BoxedUp{a}$ to each letter $a \in A$, 
such
that $\Delta$ contains 
copying transitions 
$q_0 \xrightarrow{a|a} q_0$ for all $a \in \Lambda \cup \Lambda'$, 
and 
{\setlength{\fboxsep}{2pt}$q_k \xrightarrow{\BoxedUp{a}\,|\,\BoxedUp{a}} q_k$}
 for all $a \in A$ and each $k \in \{1, \dots n \}$, as well as
erasing transitions 
$q_0 \xrightarrow{c|\,\BoxedUp{c}}q_0$ for each $c \in C$, and
$q_0 \xrightarrow{[_k|\, \Boxed{[_k}} q_k$, 
$q_k \xrightarrow{]_k|\, \Boxed{ ]_k}} q_0$ for each $k$. 
Furthermore, we require that there is a homomorphism~$h: \Lambda \to \Sigma$,
such that the remaining transitions of $\Delta$ are coding transitions
$q_0 \xrightarrow{h(a)|a} q_h$
and $q_h \xrightarrow{h(a)|a} q_h$ for all $a \in \Lambda$.
Finally, a seed~$(R, L_\acc)$ is an $n$-flower if
$R$ is an $n$-flower transducer and $L_\acc$ is a regular language
over the alphabet $A'$ of its neutralised symbols.
An example of a $2$-flower transducer is depicted in Figure~\ref{fig:flower}. 

\begin{figure}[t]
\tikzset{every picture/.style={->,>=stealth',shorten >=1pt,auto, font=\scriptsize}} 
\tikzset{initial text=, initial distance=1em} 
\tikzset{accepting by double} 
\tikzstyle{every state}=[minimum size=1em, inner sep=1.8pt, font=\scriptsize]

\begin{center}
   
  \setlength{\fboxsep}{1.1pt}
  \subfloat[$2$-flower]{
  \label{fig:flower}
  
    \begin{tikzpicture}[xscale=1,yscale=.8]
      
      \node[initial below, accepting, state]          (Q0) at (0, 0)           {$q_0$};  
      \node[state]          (Q1) at (-1.5, 1.7)          {$q_1$};
      \node[state]          (Q2) at (1.5, 1.7)           {$q_2$};
      \node[accepting, state]          (QH) at (1, - 1.5)          {$q_h$};
      
      \path (Q0) 
      edge[bend left] node[pos=0.6, outer sep=-3pt]
      {$\dfrac{(}{\noL}$}(Q1)
      edge[bend left] node[pos=0.7, outer sep=-2pt]
      {$\dfrac{[}{\noLL}$}(Q2)
      edge[bend left] (Q2)
      edge (QH)
      edge[loop, out=170, in=210, looseness=6] 
      node[below, , yshift=-3pt, outer sep=0pt]{$\dfrac{x}{x}, \dfrac{\BoxedUp{x}}{\BoxedUp{x}}
        {~\scriptstyle x \in A}$} (Q0);
      
      \path (Q1) edge[bend left ] node[pos=0.3, outer sep=-2pt]
      {$\dfrac{)}{\noR}$} (Q0)
      edge[loop above] node[above]{$\dfrac{\BoxedUp{x}}{\BoxedUp{x}}{~\scriptstyle x \in
          A}$} (Q1);
      
      \path (Q2) edge[bend left ] node[pos=0.4, outer sep=-2pt]
      {$\dfrac{]}{\noRR}$} (Q0)
      edge[loop above] node{$\dfrac{\BoxedUp{x}}{\BoxedUp{x}}{~\scriptstyle x \in
          A }$}  (Q2);
      
      \path (QH) edge[loop right]
      node{
        $\dfrac{a}{[}$, $\dfrac{b}{(}$,
        $\dfrac{c}{]}$, $\dfrac{c}{)}$
      }(QH);

    \end{tikzpicture}      
  }
\subfloat[loose flower for palindromes]{
  \label{fig:loose-flower}
 
  \begin{tikzpicture}[xscale=1,yscale=.7]
      
    \node[initial below, state]          (Q0) at (0, 0)           {$q_0$};  
    \node[state]          (Q1) at (-1, 2)          {$q_1$};
    \node[state]          (Q2) at (2, 2)           {$q_2$};
    \node[accepting, state]          (QH) at (.8, - 2)
    {$q_h$};
    \node[accepting, state]           (QF) at (1, 0) {$q_f$};
    \node[accepting, state]           (QC) at (-.8, -2) {$q_c$};
    
    \path (Q0) 
    edge[bend left] node[outer sep=-2pt, pos=0.7] 
    {$\dfrac{(}{\noL}$} (Q1)
    edge[bend left] node[outer sep=-2pt, pos=0.7] {$\dfrac{[}{\noLL}$}(Q2)
    edge (QC)
    edge (QH)
    edge[loop, out=160, in=210, looseness=5] node[below,yshift=-1pt,xshift=-10pt,outer sep=0pt]{$ \dfrac{\noL}{\noL}, \dfrac{\noLL}{\noLL}$} (Q0);
    
    \path (Q1) edge[bend left, outer sep=-2pt, pos=0.3] node
    {$\dfrac{)}{\noR}$} (QF)
    edge[loop above] node[above, outer sep=-3pt]
    {$\dfrac{x}{x} x \in A $} (Q1);
    
    \path (Q2) edge[bend left, outer sep=-2pt, pos=0.3]
    node 
    {$\dfrac{]}{\noRR}$} (QF)
    edge[loop above] node [above, outer sep=-3pt]
    {$\dfrac{x}{x}~x \in A $} (Q2);
    
    \path (QH) edge[loop right]
    node{
      $\dfrac{a}{[}$, $\dfrac{b}{(}$,
      $\dfrac{a}{]}$, $\dfrac{b}{)}$
    }(QH);
    
    \path (QC)
    edge[loop left] node {$\dfrac{\BoxedUp{x}}{\BoxedUp{x}},\dfrac{x}{x}{~\scriptstyle x \in A }$} (QC);
    
    \path (QF)
    edge[loop, out=20, in=330, looseness=5] node[below, yshift=1pt,xshift=10pt]{$ \dfrac{\noR}{\noR}, \dfrac{\noRR}{\noRR}$} (QF);
    
  \end{tikzpicture}      
}
\end{center}

\caption{Flower transducers}\label{fig:flowers}
\end{figure}
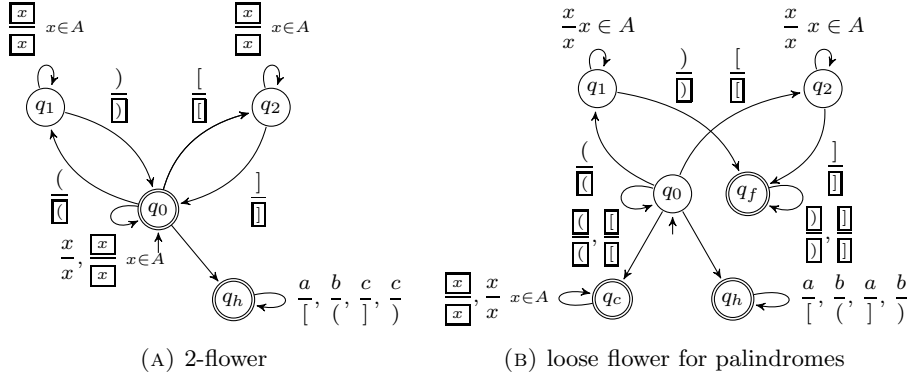

\begin{theorem}\label{thm:char-cf}
A language is context-free if, and only if,
it is covered by a consensus game acceptor with an $n$-flower seed. 
\end{theorem}

\begin{proof}
Let $L \subseteq \Sigma^*$ be an arbitrary context-free language. 
According to the Representation Theorem~\ref{thm:Okhotin}, there exists 
an alphabet $\Lambda$ partitioned into a bracket alphabet~$B_n$ 
and a set of neutral symbols $C$, a letter-to-letter homomorphism given by 
$h : \Lambda \to \Sigma$, and a regular language $M
\subseteq \Lambda^*$, such that $L = h( D_{\Lambda} \cap M )$. 
With these in hand, we construct the seed transducer $R_{n, C}$ for
the Dyck language~$D_{\Lambda}$ as in Lemma~\ref{lem:Dyck} and 
then add a coding cycle~$q_h$ for the homomorphism~$h$.
According to Lemma~\ref{lem:hom-reg}, the transducer~$R$ constructed in this way
together with the seed language $L_\acc := \Box^*\times M$ 
describe a game that covers $L$. 
By reducing~$R$ to the alphabet $\Sigma \cup A \cup \{ \Box \} \times
A$, we finally obtain the~$n$-flower seed 
seed~$(\hat{R}, \Box^*\times M)$ that also covers~$L$ 
over~$\Sigma$.

For the converse statement, let $G$ be a consensus game
with an $n$-flower seed $(R, L_\acc)$ over an alphabet~$\Gamma$. We
wish to prove that the language covered by~$G$ over~$\Gamma$ is
context-free. First, 
we partition~$\Gamma$ into an alphabet~$\Sigma$, a set~$B_n$ of $n$ matching
brackets, a set $C$ of (prime) neutral symbols, and a neutralised copy~$\Lambda'$ of $\Lambda := B_n
\cup C$. Let~$\DA$ be the Dyck language over~$\Lambda$, and let
$h: \Lambda \to \Sigma$ be the letter-to-letter homomorphism
determined by the coding cycle in~$R$. Then, 
consider the \emph{neutralising} letter-to-letter homomorphism
$\nu: A \cup A' \to A'$ which
maps both $a$ and $\BoxedUp{a}$
to the neutralised copy $\BoxedUp{a}$, and set $M := \nu^{-1} L_\acc$; 
as an inverse homomorphic image of a regular language, $M$ is
regular.
By construction of the flower transducer, 
we know that, over the
alphabet~$\Sigma$, the game $G$ covers the context-free language $h( \DA \cap M)$.
 
To describe the language covered over~the
full alphabet~$\Gamma$, consider the Dyck language
$\DA'$ over the alphabet~$A \cup A'$ with the same set~$B_n$ of
brackets as $\DA$, but with an extended set~$C \cup A'$ of neutral
symbols. We observe that $\DA'$ is closed under $R$ and $R^{-1}$ in
the sense that, for any pair of non-terminal words $w, w' \in \Gamma^* \setminus
\Sigma^*$ with $(w, w') \in R$, we have $w \in \DA'$ if, and
only if, $w' \in \DA'$. Moreover, $\nu( w ) = \nu( w' )$
(letters are neutralised, but never forgotten). 
This implies that, over $\Gamma \setminus \Sigma$, the game $G$ covers
the language $\DA' \cap M$. 
Since every observation sequence of~$G$ is either contained in $\Sigma^*$ or in $(\Gamma
\setminus \Sigma)^*$, it follows that, 
over the full alphabet~$\Gamma$, the 
consensus game~$G$ covers the context-free language $h( \DA \cap M ) \cup (\DA'
\cap M )$. 
\end{proof}

\setlength{\fboxsep}{1pt}
Notice that without restricting the alphabet of the accepting seed language to
neutralised symbols, the flower structure of the
transducer alone would not guarantee that the language covered by a
game is context-free. For instance, the one-flower transducer 
over a bracket pair $[,]$ one neutral symbol $\#$, and their neutralised
copies, together with the seed
language $L_\acc := [^+\, \# \, \BoxedUp{[}^+\, ]^+$ give rise to a game where
the covered language~$L$ is not context-free, since the intersection $L \cap \BoxedUp{[}^+\#\, [^+\, ]^+ =  
\BoxedUp{[}^n\#\, [^n\, ]^n~$ is not context-free.

Returning to games, the argument from Lemma~\ref{lem:char-winning} shows that, given a game with seed $(R, L_\acc, L_\rej)$, 
at every play~$\pi$ with observation~$\beta^1( \pi )$  
in the language~$L_1$ covered by $(R, L_\acc)$ over the full
observation alphabet, the only safe decision is~$1$, whereas 
at each play with observations in the language~$L_0$ covered by $(R,
L_\rej)$ the only safe decision is~$0$. 
An observation-based strategy prescribing~$s^1( \pi ) :=1$ 
precisely if $\beta^1(\pi) \in L_1$ induces a joint strategy that is
\emph{optimal} in the sense that it prescribes a safe decision whenever one
exists. Likewise, a strategy that prescribes~$0$ 
precisely at sequences in $L_0$ is optimal. Optimal strategies are
\emph{undominated}, that is, no other strategy wins
strictly more plays. Clearly, if a game is solvable, 
then every optimal strategy is winning.

One consequence of Theorem~\ref{thm:char-cf} is that, for games where
one of $(R, L_\acc)$ or $(R, L_\rej)$ is an
$n$-flower, the set~$L_1$ or~$L_0$ is context-free, and therefore
recognisable by a nondeterministic push-down automata, which we can
construct effectively from the game description. The obtained
automaton hence implements an optimal strategy that is indeed a
winning strategy, in case the game is solvable.
According to Theorem~\ref{thm:char-iterated}, already for games where both $L_0$ and $L_1$ are context-free,
the question of whether a winning strategy exists amounts to solving the 
disjointness problem for context-free languages and is hence
undecidable. Under these circumstances, it is remarkable that we can
effectively construct strategies that are optimal and, moreover, winning
whenever the game is solvable.

\begin{corollary}
For any consensus game~$G$ with seed $(R, L_\acc, L_\rej)$
where either $(R, L_\acc)$ or $(R, L_\rej)$ is an $n$-flower, we can effectively construct a
push-down automaton~$\calS$ that implements an optimal strategy.
\end{corollary}

\section{Conclusion}

We presented a simple kind of games 
with imperfect information where constructing optimal
strategies requires iterating the (synchronous rational) relation that
correlates the observation of players. This establishes a correspondence between winning strategies in games on
the one hand, and main classes of formal languages on the other hand.
The correspondence leads to several insights on games with imperfect
information. 

Firstly, we obtain simple examples that illustrate
the computational complexity of coordination under imperfect information.
The classical constructions for 
proving that the problem is undecidable in the general case, 
typically involve an unbounded number of 
non-trivial decisions by which the players describe
configurations of a Turing
machine~\cite{PnueliRos90,AzharPetRei01,Schewe2014}.
In contrast, our 
undecidability argument in Theorem~\ref{thm:undecidable} relies on a single simultaneous
decision.  

Secondly, we identify families of games where
optimal strategies exist and can be constructed effectively, 
but the complexity of the strategic decision necessarily grows 
with the length of the play. 
This opens a new perspective for distributed strategy synthesis that
departs from the traditional focus on finite-state
winning strategies and on game classes on which the existence of such is decidable.
In consensus games, the implementation of winning strategies 
requires arbitrary linear-bounded automata in the general
case. However, we also described a structural condition on game
graphs that ensures that winning strategies can be
implemented by push-down automata. 

One challenging objective is to
classify games with imperfect information according to the complexity of
strategies required for solving them.
The insights developed for consensus games allow a few more steps
into this direction. For instance, 
games with one-flower seeds cover
one-counter languages and therefore admit optimal strategies
implemented by one-counter automata. Likewise, we can build up
a variant of $n$-flower seeds from Dyck languages
restricted to palindromes, as illustrated in
Figure~\ref{fig:loose-flower}. Games with seeds of this kind
cover a subclass of linear
languages and hence admit
optimal strategies implemented by one-turn push-down automata.

{
\bibliographystyle{amsplain}
\bibliography{cs}
}

\end{document}